\documentclass[aps,pra,twocolumn,superscriptaddress,showpacs,amsmath,amssymb]{revtex4-2}

\usepackage{graphicx}
\usepackage{bm}
\usepackage{braket}
\usepackage{booktabs}
\usepackage{amsthm}
\usepackage{hyperref}

\newtheorem{theorem}{Theorem}
\newtheorem{proposition}{Proposition}
\newtheorem{lemma}{Lemma}
\theoremstyle{remark}
\newtheorem{remark}{Remark}

\newcommand{\Dket}{\ket{D_N^{(2)}}}
\newcommand{\R}{\mathbb{R}}
\newcommand{\Z}{\mathbb{Z}}
\newcommand{\Heff}{H_{\mathrm{eff}}}

\begin{document}

\title{Perfect State Transfer from a Localised Two-Excitation State\\
to a Dicke State via Static Spin-Network Hamiltonians}

\author{Soumyojyoti Dutta}
\affiliation{A.~P.~Shah Institute of Technology, Ghodbunder Road,
Kasarvadavali, Thane West, Thane, Maharashtra 400615, India}
\email{m24iqt014@alumni.iitj.ac.in}

\date{September 9, 2026}

\begin{abstract}
I construct a family of time-independent, excitation-preserving spin Hamiltonians
that realises perfect state transfer from a localised two-excitation state to the
symmetric two-excitation Dicke state, for every system size $N \ge 4$.
The Hamiltonian has the physical form
$H=\sum_{i<j}J_{ij}(\sigma_i^+\sigma_j^-+\sigma_j^+\sigma_i^-)+\sum_i\epsilon_i n_i$
with real couplings, and satisfies
$e^{-iHt}\ket{110\cdots0}=e^{i\phi}\Dket$ at a finite time.
The construction exploits an $S_{N-2}$ permutation symmetry acting on the initially
unoccupied spins, which reduces the dynamics to a four-dimensional invariant subspace.
Requiring $\tfrac12(\ket{\psi_0}+\Dket)$ to be a zero eigenvector determines the
on-site energies in closed form and leaves three coupling parameters free.
The remaining inverse spectral problem reduces to two polynomial equations in two
dimensionless coupling ratios; eliminating one ratio yields a degree-six reciprocal
polynomial, which the substitution $z=x+x^{-1}$ converts to a cubic.
I then fix the spectral family to $(-n,-1,1)$ with $n$ an odd integer. Factorising
the leading coefficient of the cubic and evaluating it at the boundary $z=-2$ shows
that for every $N\ge4$ some odd $n$ is large enough to produce a real root below
$-2$, and a subresultant supplies a real lift of that root back to the original
system. No numerical optimisation enters the existence argument, which is symbolic
throughout; I propagate the resulting Hamiltonians only to check the algebra.
The coefficient in question happens to contain no odd powers of $n$, so the required
$n$ comes with an explicit threshold. The construction can then be costed: couplings
grow as $N^{1/2}$ and the on-site range as $N^{3/2}$, the transfer time stays within
roughly a factor of three of the Mandelstam--Tamm limit at every size, and the
fidelity proves sensitive to systematic drift of the spectator--spectator coupling
class while tolerating independent bond disorder, which self-averages.
The result is a constrained analogue of perfect state transfer: unlike the general
real-state problem, where an unconstrained real symmetric matrix suffices, here the
Hamiltonian is required to arise from an excitation-preserving spin-network form.
\end{abstract}

\maketitle

\section{Introduction}

Perfect state transfer (PST) in quantum spin networks concerns the existence of a
fixed Hamiltonian whose evolution carries an initial state exactly to a prescribed
final state at a finite time. Engineered static interactions are the central
mechanism, and the theory of PST in engineered spin chains is well
developed~\cite{Christandl2004,Kay2010}.

From the spectral point of view, exact transfer between real states is also well
understood. Strong cospectrality and its generalisations provide the appropriate
language for characterising the eigenspace projections of the states
involved~\cite{Godsil2012,Coutinho2014}; an analogous characterisation for
discrete-time walks appears in Ref.~\cite{ChanZhan2023}. Recent work has extended this
framework to arbitrary real pure states; in particular, for a pair of real pure states
one may construct a real symmetric matrix realising perfect transfer between
them~\cite{GKM2025}.

That result does not settle the physical question asked here. An unconstrained
construction returns a matrix that is generically dense and carries no reason to have
the excitation-preserving hopping structure
\begin{equation}
H=\sum_{i<j}J_{ij}\bigl(\sigma_i^+\sigma_j^-+\sigma_j^+\sigma_i^-\bigr)
 +\sum_i \epsilon_i n_i ,
\label{eq:model}
\end{equation}
or to respect the coupling and symmetry constraints a spin network imposes. The
point is easy to make concrete. Given real unit vectors $e_0,v$ and any $t$, the
Householder reflection $R=\mathbb{1}-2ww^{T}/\|w\|^{2}$ with $w=e_0-v$ already sends
$e_0$ to $v$, and $H=(\pi/t)(\mathbb{1}-R)/2$ gives $e^{-iHt}=R$. That $H$ is dense
and will not usually take the form~\eqref{eq:model}. So the open question is whether a
\emph{physically constrained} excitation-preserving Hamiltonian can do the same job,
not whether some real symmetric matrix can.

I answer this affirmatively for the transfer
\begin{equation}
\ket{110\cdots0}\;\longrightarrow\;\Dket ,\qquad
\Dket=\binom{N}{2}^{-1/2}\!\!\sum_{1\le i<j\le N}\!\!\ket{1_i1_j} .
\end{equation}
The initial state is deliberately \emph{not} permutation symmetric: its two
excitations sit on two distinguished sites, while the target is fully symmetric. That
asymmetry is what separates the problem from Krawtchouk and $\mathrm{SU}(2)$
constructions (Sec.~\ref{sec:krawtchouk}), where both states live inside a single
symmetric representation. What survives here is only $S_{N-2}$ acting on the
initially unoccupied sites, and having just that much symmetry turns out to be the
useful case: enough to reduce the problem, not so much that the two occupied sites
become interchangeable.

\begin{figure*}[t]
\centering
\includegraphics[width=\textwidth]{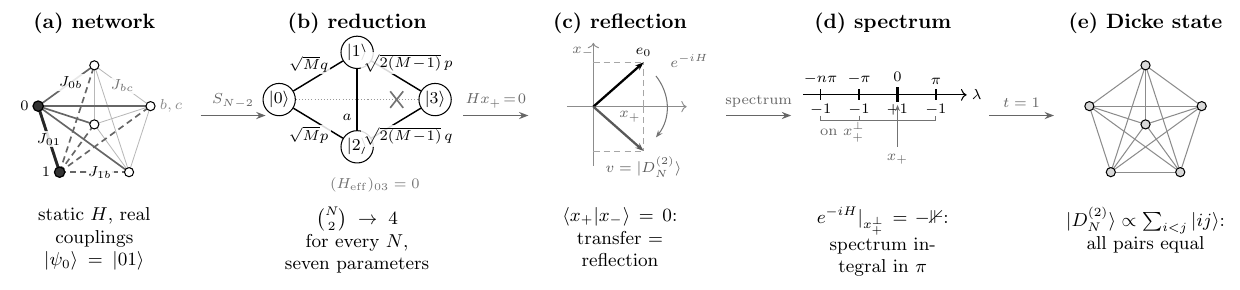}
\caption{\label{fig:concept}%
The construction, end to end. (a)~$N$ spins with real excitation-preserving
couplings and on-site energies, starting from $\ket{\psi_0}=\ket{01}$, which puts
both excitations on a single bond. (b)~$S_{N-2}$ acting on the $M=N-2$ initially
unoccupied spins leaves a four-dimensional invariant subspace that holds both
$\ket{\psi_0}$ and $\Dket$, so the $\binom N2$-dimensional problem becomes a
$4\times4$ one at every $N$. The entry $(\Heff)_{03}$ vanishes because $\ket{01}$ and
$\ket{bc}$ differ by two hops. (c)~Writing $e_0=x_++x_-$ and $v=x_+-x_-$ with
$\braket{x_+|x_-}=0$ turns the transfer into a reflection in the $x_+$ axis.
(d)~That reflection comes from giving $e^{-iH}$ eigenvalue $+1$ on $x_+$ and $-1$ on
all of $x_+^\perp$, which forces an integral spectrum in units of $\pi$ and so imposes
two polynomial conditions on the couplings. (e)~At $t=1$ every pair of spins carries
the same amplitude. These panels are schematic; the numbers live in
Figs.~\ref{fig:reduction},~\ref{fig:scaling},~\ref{fig:dynamics}
and~\ref{fig:disorder}.}
\end{figure*}

\section{Model and notation}

I consider $N$ qubits with the Hamiltonian~\eqref{eq:model},
$J_{ij},\epsilon_i\in\R$, $n_i=\sigma_i^+\sigma_i^-$. Since
$[H,\hat N]=0$ with $\hat N=\sum_i n_i$, the dynamics decomposes into sectors of
fixed Hamming weight. We work in the two-excitation sector $\mathcal H_2$,
$\dim\mathcal H_2=\binom N2$, with
\begin{equation}
\ket{\psi_0}=\ket{110\cdots0},\qquad
\Dket=\binom N2^{-1/2}\sum_{i<j}\ket{ij},
\end{equation}
and I seek
\begin{equation}
e^{-iHt}\ket{\psi_0}=e^{i\phi}\Dket .
\label{eq:goal}
\end{equation}
Rescaling $H$ is equivalent to inversely rescaling $t$, so I set $t=1$ throughout;
$t$ is recovered at the end by rescaling all couplings.

\subsection{The symmetric family}

Label the initially occupied sites $0,1$ and set
\begin{equation}
B=\{2,\dots,N-1\},\qquad M=N-2 .
\end{equation}
We impose invariance under $S_M=S_{N-2}$ permuting $B$. The most general
Hamiltonian of the form~\eqref{eq:model} with this symmetry has \emph{seven}
parameters, independent of $N$:
\begin{equation}
\begin{aligned}
a&=J_{01}, &
p&=J_{0b}\ \ (b\in B), &
q&=J_{1b}\ \ (b\in B),\\
w&=J_{bc}\ \ (b\ne c\in B), &
\epsilon_0,\ \epsilon_1,\ \epsilon_b&\ \ (b\in B). &&
\end{aligned}
\label{eq:params}
\end{equation}
Nothing here is geometrically local. The interaction graph is a weighted complete
graph sorted into the symmetry classes above, with a single intra-$B$ coupling $w$.
The constraints that do the work are excitation-number conservation and the
seven-parameter $S_{N-2}$-symmetric structure; distance on a lattice plays no role.

Observe that $p$ and $q$ stay distinct, since $S_M$ does \emph{not} exchange sites
$0$ and $1$. Enlarging the symmetry to $S_2\times S_M$, which means $p=q$ and
$\epsilon_0=\epsilon_1$, equalises the two off-diagonal blocks of the reduced matrix
and kills the construction outright. The gap between $p$ and $q$ is doing real work.

\section{Symmetry reduction}

The $S_M$-invariant subspace of $\mathcal H_2$ is four-dimensional, spanned by
\begin{equation}
\begin{aligned}
\ket{0}&=\ket{01}, &
\ket{1}&=\tfrac1{\sqrt M}\textstyle\sum_{b\in B}\ket{0b},\\
\ket{2}&=\tfrac1{\sqrt M}\textstyle\sum_{b\in B}\ket{1b}, &
\ket{3}&=\binom M2^{-1/2}\!\!\textstyle\sum_{b<c\in B}\!\!\ket{bc}.
\end{aligned}
\label{eq:orbits}
\end{equation}
In this basis $\ket{\psi_0}=e_0=(1,0,0,0)^T$ and
\begin{equation}
v \;=\;\binom{M+2}{2}^{-1/2}
\Bigl(1,\ \sqrt M,\ \sqrt M,\ \sqrt{\tfrac{M(M-1)}2}\Bigr)^{T}.
\label{eq:v}
\end{equation}
Both $e_0$ and $v$ lie in this subspace, and it is invariant under $H$, so the entire
problem is four-dimensional for every $N$.

Writing the parameters~\eqref{eq:params} in the basis~\eqref{eq:orbits} gives
\begin{widetext}
\begin{equation}
\Heff=
\begin{pmatrix}
d_0 & \sqrt M\,q & \sqrt M\,p & 0\\
\sqrt M\,q & d_1 & a & \sqrt{2(M-1)}\,p\\
\sqrt M\,p & a & d_2 & \sqrt{2(M-1)}\,q\\
0 & \sqrt{2(M-1)}\,p & \sqrt{2(M-1)}\,q & d_3
\end{pmatrix},
\label{eq:Heff}
\end{equation}
\end{widetext}
with
\begin{equation}
\begin{aligned}
d_0&=\epsilon_0+\epsilon_1, &
d_1&=\epsilon_0+\epsilon_b+(M-1)w,\\
d_2&=\epsilon_1+\epsilon_b+(M-1)w, &
d_3&=2\epsilon_b+2(M-2)w .
\end{aligned}
\label{eq:diag}
\end{equation}
The structural constraints inherited from the spin-network form are
\begin{equation}
(\Heff)_{03}=0,\qquad
\frac{(\Heff)_{01}}{(\Heff)_{23}}
=\frac{(\Heff)_{02}}{(\Heff)_{13}}
=\sqrt{\frac{M}{2(M-1)}} .
\label{eq:constraints}
\end{equation}
The reduction alone would leave an ordinary $4\times4$ inverse eigenvalue problem;
it is these constraints that make it something else. The vanishing entry has a simple
reading: $\ket{01}$ and $\ket{bc}$ differ by two hops, and $H$ moves one excitation
at a time.

\begin{figure*}[t]
\centering
\includegraphics[width=0.72\textwidth]{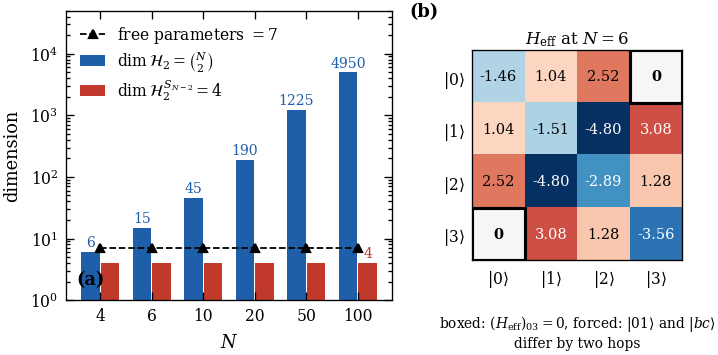}
\caption{\label{fig:reduction}%
\emph{(a)}~$\dim\mathcal H_2=\binom N2$ grows without bound, while the
$S_M$-invariant subspace holding $\ket{\psi_0}$ and $\Dket$ stays four-dimensional
and the parameter count stays at seven. An all-$N$ statement is tractable for that
reason. \emph{(b)}~The resulting $\Heff$ at $N=6$ in the orbit
basis~\eqref{eq:orbits}, for the solution of Sec.~\ref{sec:worked}. The boxed entry
is the structural zero of~\eqref{eq:constraints}, which holds for every member of the
family and not only for this solution.}
\end{figure*}

\begin{lemma}[Physical realisability]
\label{lem:realisable}
For every $M\ge2$ the map
$(\epsilon_0,\epsilon_1,\epsilon_b,w)\mapsto(d_0,d_1,d_2,d_3)$ defined
by~\eqref{eq:diag} is a linear bijection, with constant Jacobian determinant $4$ and
inverse
\begin{equation}
\begin{aligned}
\epsilon_0&=\tfrac12(d_0+d_1-d_2), &
\epsilon_1&=\tfrac12(d_0-d_1+d_2),\\
w&=\tfrac12(-d_0+d_1+d_2-d_3), &
\epsilon_b&=\tfrac12 d_3-(M-2)\,w .
\end{aligned}
\end{equation}
\end{lemma}

Lemma~\ref{lem:realisable} is what keeps the construction physical. Whatever
diagonal $(d_0,d_1,d_2,d_3)$ the argument below selects, genuine on-site energies and
a genuine intra-$B$ coupling $w$ produce it. Generically $w\ne0$, so the spectator
sites cannot be left mutually uncoupled.

\section{Exchange symmetry}

Let $S$ be the orthogonal matrix exchanging orbit states $\ket1$ and $\ket2$. Direct
computation from~\eqref{eq:Heff}--\eqref{eq:diag} gives
\begin{equation}
S\,H(p,q,a)\,S^{T}=H(q,p,a),\qquad Se_0=e_0,\quad Sv=v .
\label{eq:exchange}
\end{equation}
Physically, $S$ swaps the two initially occupied sites. Neither the Dicke state nor
$e_0$ can tell them apart, so with
\begin{equation}
x_\pm=\tfrac12(e_0\pm v),\qquad \braket{x_+|x_-}=0,
\label{eq:xpm}
\end{equation}
one gets $Sx_\pm=x_\pm$, and the whole transfer geometry is $\Z_2$ invariant under
$p\leftrightarrow q$. This is where the reciprocal structure of Sec.~\ref{sec:cubic}
comes from. A physical symmetry forces it; the elimination merely reveals it.

The orthogonality asserted in~\eqref{eq:xpm} holds even though the source and target
are \emph{not} orthogonal: $\braket{e_0|v}=\binom N2^{-1/2}\ne0$. For real normalised
$e_0,v$ the overlap is real, so the cross terms cancel,
\begin{equation}
\braket{e_0+v|e_0-v}=\|e_0\|^2-\braket{e_0|v}+\braket{v|e_0}-\|v\|^2=0 .
\end{equation}
So~\eqref{eq:xpm} is available for any pair of real normalised states, and the
mechanism of Sec.~\ref{sec:stage1} never needs orthogonal endpoints. I spell this out
because the familiar site-to-site PST setting does have $\braket{e_0|v}=0$, which
makes it easy to read the general argument as depending on that.

\section{Spectral formulation and Stage I}
\label{sec:stage1}

From~\eqref{eq:xpm}, $e_0=x_++x_-$ and $v=x_+-x_-$. Suppose $H$ is chosen so that
\begin{equation}
Hx_+=0,
\label{eq:stage1}
\end{equation}
and so that the three eigenvalues of $\Heff|_{x_+^{\perp}}$ are
\begin{equation}
\{-n\pi,\,-\pi,\,\pi\},\qquad n\ \text{odd}.
\label{eq:spectrum}
\end{equation}
Since $0$ is then a simple eigenvalue, $x_-\in x_+^{\perp}$ lies entirely in the span
of the remaining eigenvectors. For odd $n$, $e^{in\pi}=(-1)^n=-1$ and
$e^{\mp i\pi}=-1$, so \emph{every} eigenvalue in~\eqref{eq:spectrum} contributes the
same phase $-1$. Hence $e^{-iH}$ acts as the identity on $\mathrm{span}\{x_+\}$ and as
$-\mathbb 1$ on the whole of $x_+^{\perp}$:
\begin{equation}
e^{-iH}x_+=x_+,\qquad e^{-iH}x_-=-x_- ,
\end{equation}
whence
\begin{equation}
e^{-iH}e_0=x_+-x_-=v .
\label{eq:transfer}
\end{equation}
Why do~\eqref{eq:stage1} and~\eqref{eq:spectrum} suffice by themselves? Because
$e^{-iH}=-\mathbb 1$ on the whole of $x_+^{\perp}$, the second relation needs nothing
more than $x_-\in x_+^{\perp}$. In particular $x_-$ need \emph{not} be an eigenvector
of $H$, and no condition is placed on the eigenvectors at all beyond the spectrum on
$x_+^\perp$. Everything now rests on realising~\eqref{eq:stage1}
and~\eqref{eq:spectrum} inside the family~\eqref{eq:Heff}.

\begin{proposition}[Stage-I freedom]
\label{prop:stage1}
Set $c=\binom N2^{-1/2}$ and $\gamma=(1+c)/c$. For \emph{any} $p,q,a$, the choice
\begin{equation}
\begin{aligned}
d_0&=-\frac{M(p+q)}{\gamma}, &
d_1&=-a-\gamma q-(M-1)p,\\
d_3&=-2(p+q), &
d_2&=-a-\gamma p-(M-1)q,
\end{aligned}
\label{eq:stage1diag}
\end{equation}
makes $x_+$ a zero eigenvector of $\Heff$. In particular $\det\Heff=0$.
\end{proposition}

\begin{proof}
$x_+\propto\bigl(1+c,\ c\sqrt M,\ c\sqrt M,\ c\sqrt{M(M-1)/2}\bigr)$.
The four components of $\Heff x_+=0$ are linear in $d_0,\dots,d_3$ with nonvanishing
coefficients, and~\eqref{eq:stage1diag} is their unique solution. Two identities
simplify the result: $\sqrt{2(M-1)}\,X_3=(M-1)X_1$ and
$\sqrt{2(M-1)}\,X_1/X_3=2$, where $X_k$ are the components of $x_+$.
\end{proof}

Proposition~\ref{prop:stage1} is the structural fact everything downstream leans on.
Demanding that $x_+$ be a zero eigenvector turns out to cost nothing, since it places
\emph{no} condition on $p,q,a$, and Lemma~\ref{lem:realisable} guarantees the
diagonal it produces is physically realisable. What is left to arrange is the
spectrum on $x_+^\perp$, and that is where the difficulty now lives.

\section{The inverse spectral problem}

The matrix defined by~\eqref{eq:Heff} and~\eqref{eq:stage1diag} is homogeneous,
$H(cp,cq,ca)=cH(p,q,a)$, so the eigenvalue \emph{ratios} depend only on
\begin{equation}
x=q/p,\qquad y=a/p .
\end{equation}

\begin{lemma}[The chart $p\ne0$ suffices]
\label{lem:chart}
Setting $p=0$ in the spectral equations~\eqref{eq:F},~\eqref{eq:G} below yields only
$a=q=0$, hence $H=0$, which effects no transfer. By~\eqref{eq:exchange} the same
holds for $q=0$. Every nontrivial solution therefore has $p\ne0$, and one may set
$p=1$.
\end{lemma}

\begin{proof}
With $p=0$, Eqs.~\eqref{eq:F}--\eqref{eq:G} reduce to two forms in $(a,q)$,
homogeneous of degrees $2$ and $3$ respectively. Over the field
$\mathbb Q(M,\gamma,n)$ a lexicographic Gr\"obner basis of the ideal they generate,
with $a>q$, contains the monomial $q^{4}$; hence $q=0$ on the common zero set, and
the degree-two generator then forces $a^{2}=0$. Equivalently, putting $u=a/q$ and
eliminating $u$, a solution with $q\ne0$ would require the vanishing of a resultant
$\mathcal R(M,\gamma,n)$, a polynomial of degree $4$ in $n$; $\mathcal R$ has no root
at any positive odd integer $n$ for $2\le M\le100$ with the physical value
$\gamma=1+\sqrt{(M+1)(M+2)/2}$, so the generic conclusion holds at every parameter
value used here. Both computations appear in the companion notebook.
\end{proof}

Let $e_1,e_2,e_3$ be the elementary symmetric functions of the three nonzero
eigenvalues (equivalently, of the four eigenvalues of $\Heff$, one of which is $0$).
For a spectral triple $(m_1,m_2,m_3)$ put $s_1=\sum m_i$,
$\sigma_2=\sum_{i<j}m_im_j$, $\sigma_3=m_1m_2m_3$. The eigenvalues are proportional to
$(m_1,m_2,m_3)$ if and only if
\begin{align}
F&\equiv\sigma_2 e_1^2-s_1^2e_2=0,\label{eq:F}\\
G&\equiv\sigma_3 e_1^3-s_1^3e_3=0,\label{eq:G}
\end{align}
the overall scale being fixed afterwards by $e_1=s_1\pi$. That this rescaling is
legitimate requires $e_1\ne0$, which the following lemma establishes.

\begin{lemma}[Nondegeneracy of the scale]
\label{lem:e1}
For the spectral triple $(m_1,m_2,m_3)=(-n,-1,1)$ with $n$ a positive odd integer,
any solution of~\eqref{eq:F}--\eqref{eq:G} with $p\ne0$ has $e_1\ne0$.
\end{lemma}

\begin{proof}
Suppose $e_1=0$. Since $s_1=-n\ne0$, Eq.~\eqref{eq:F} reduces to $-s_1^2e_2=0$,
so $e_2=0$; Eq.~\eqref{eq:G} likewise gives $e_3=0$. Let
$\lambda_1,\lambda_2,\lambda_3$ be the eigenvalues of $\Heff$ other than the zero
eigenvalue supplied by Proposition~\ref{prop:stage1}. Then
\begin{equation}
e_2=\tfrac12\Bigl[\bigl(\textstyle\sum_i\lambda_i\bigr)^2-\sum_i\lambda_i^2\Bigr]
   =-\tfrac12\sum_i\lambda_i^2 ,
\end{equation}
using $e_1=\sum_i\lambda_i=0$. Hence $e_2=0$ forces
$\lambda_1=\lambda_2=\lambda_3=0$, so all four eigenvalues of $\Heff$ vanish and,
$\Heff$ being real symmetric, $\Heff=0$. But $(\Heff)_{02}=\sqrt M\,p\ne0$ for
$p\ne0$, a contradiction.
\end{proof}

\section{Reciprocal sextic and cubic reduction}
\label{sec:cubic}

Eliminating $y$ gives
\begin{equation}
\mathrm{Res}_y(F,G)=\gamma^{6}\,P(x),\qquad \deg_x P=6,
\label{eq:res}
\end{equation}
\emph{Normalisation.} A resultant is defined only up to a nonzero factor, so
Eq.~\eqref{eq:res} does not by itself determine $P$: any $P\mapsto cP$ with
$c=c(M,\gamma)\in\R\setminus\{0\}$ is equally valid, and the explicit values quoted
below refer to one fixed choice, namely the primitive degree-six factor returned by
the elimination. What the argument actually uses is invariant under such a
rescaling: $\deg_xP$, the palindromic property, the location of the roots, and the
sign of the product $A\cdot P(-1)$, where $A$ is the leading coefficient (both $A$ and
$P(-1)$ scale by $c$, so their product scales by $c^{2}>0$). Individual values of $A$
and of $P(-1)$ are not invariant, and are quoted only in the fixed normalisation.
The polynomial $P$ is palindromic for arbitrary $s_1,\sigma_2,\sigma_3$,
\begin{equation}
P(x)=x^{6}P(1/x),
\end{equation}
as required by the exchange symmetry~\eqref{eq:exchange}. Writing
$P(x)=\sum_{k=0}^{6}A_k^{P}x^{k}$ with $A^P_k=A^P_{6-k}$, the substitution
$z=x+x^{-1}$ gives $P(x)=x^{3}Q(z)$ with
\begin{equation}
Q(z)=A^P_6z^{3}+A^P_5z^{2}+(A^P_4-3A^P_6)z+(A^P_3-2A^P_5).
\end{equation}
Real $x\ne0$ corresponds precisely to $|z|\ge2$. We use the branch $z<-2$, whose
boundary $z=-2$ is $x=-1$, i.e.\ $q=-p$.

\begin{remark}[Why this branch, and what else is there]
\label{rem:branch}
The cubic $Q$ may also have a real root with $z>+2$, which is equally admissible:
the requirement is $|z|\ge2$, not a sign. The branch $z<-2$ is singled out because it
is the one whose existence can be established uniformly in $N$ by the sign argument of
Sec.~\ref{sec:exist}, and because it is the only admissible branch at the smallest
system size: for $N=4$ the three roots of $Q$ are $z\simeq-2.422,\,-1.757,\,+1.077$,
of which only the first satisfies $|z|\ge2$. For $N\ge6$ a root with $z>+2$ is
present as well and yields a second, inequivalent family of Hamiltonians realising the
same transfer. Within each branch the two representatives
$x_*$ and $1/x_*$ of~\eqref{eq:xstar} are exchanged by $p\leftrightarrow q$ and give
the same physics with $\epsilon_0\leftrightarrow\epsilon_1$. The worked example of
Sec.~\ref{sec:worked} is taken from the $z>+2$ family.
\end{remark}

\begin{lemma}[Boundary value]
\label{lem:boundary}
In the normalisation fixed above, and for all $M\ge2$, $\gamma>0$, $\sigma_3\ne0$,
\begin{equation}
P(-1)=8\gamma^{6}\sigma_3^{2}
\bigl[(M-\gamma)^{2}+4M+2\gamma-3\bigr]^{3}>0 ,
\label{eq:Pm1}
\end{equation}
so that $Q(-2)=-P(-1)<0$.
\end{lemma}

The bracket is positive because $(M-\gamma)^2\ge0$ and $4M+2\gamma-3>0$ for
$M\ge2$, $\gamma>0$; the prefactor $8\gamma^6\sigma_3^2$ is positive because
$\gamma>0$ and $\sigma_3\ne0$. (For the family~\eqref{eq:family} used below,
$\sigma_3=n\ne0$ since $n$ is odd.) Under $P\mapsto cP$ both sides of~\eqref{eq:Pm1}
scale by $c$, so the strict inequality is stated in the fixed normalisation; the
combination used in Sec.~\ref{sec:exist} is invariant.

\section{Existence of a real root}
\label{sec:exist}

We now specialise to the family
\begin{equation}
(m_1,m_2,m_3)=(-n,-1,1),\qquad
s_1=-n,\ \ \sigma_2=-1,\ \ \sigma_3=n,
\label{eq:family}
\end{equation}
with $n$ a positive odd integer. Write $A\equiv A^P_6$ for the leading coefficient,
which is also the leading coefficient of $Q$. For the family~\eqref{eq:family}, $A$ is
a quartic in $n$, $A=\sum_{k=0}^{4}A_kn^{k}$.

\begin{lemma}[Sign of the leading coefficient]
\label{lem:a4}
\begin{equation}
A_4=-\gamma^{2}(M+1)(M+\gamma^{2})(M^{2}+3M+2\gamma^{2})
\bigl(M\gamma+2M+\gamma^{2}+3\gamma\bigr)^{2}<0
\end{equation}
for all $M\ge2$ and $\gamma>0$, in the same fixed normalisation.
\end{lemma}

Every factor is manifestly positive, so the inequality never needs the relation
between $\gamma$ and $M$. Hence $A\to-\infty$ as $n\to\infty$, and every $M$ admits
arbitrarily large odd $n$ with $A<0$. Together with
Lemma~\ref{lem:boundary} this gives $A\cdot P(-1)<0$, equivalently
$A\cdot Q(-2)>0$. That statement is invariant under $P\mapsto cP$, and it is the only
consequence of the two lemmas the argument actually needs. For such $n$,
$Q(z)\to+\infty$ as $z\to-\infty$ while $Q(-2)<0$, so by the intermediate value
theorem there is a real root $z_*<-2$ (Fig.~\ref{fig:cubic}), and
\begin{equation}
x_*=\tfrac12\bigl(z_*\pm\sqrt{z_*^{2}-4}\bigr)\in\R .
\label{eq:xstar}
\end{equation}

\begin{remark}[Explicit threshold in $n$]
\label{rem:threshold}
For the family~\eqref{eq:family} the odd powers of $n$ in $A$ vanish identically,
$A_1=A_3=0$, so $A$ is a \emph{quadratic in $n^{2}$},
\begin{equation}
A(n)=A_4n^{4}+A_2n^{2}+A_0 ,
\end{equation}
with $A_4<0$ by Lemma~\ref{lem:a4} and, in the same normalisation,
\begin{equation}
A_0=-2\gamma^{2}(M+\gamma)^{2}(\gamma+1)^{2}\bigl(M^{2}+3M+2\gamma^{2}\bigr)^{2}<0 .
\end{equation}
So $A<0$ is more than an asymptotic statement. Put $t=n^{2}$; the condition can fail
only between the roots of $A_4t^{2}+A_2t+A_0$, and only when those are real and
positive. Therefore
\begin{equation}
n>\sqrt{t_+},\qquad
t_+=\frac{-A_2-\sqrt{A_2^{2}-4A_4A_0}}{2A_4}
\label{eq:threshold}
\end{equation}
suffices, and ``sufficiently large odd $n$'' may be replaced throughout by the
explicit threshold
$n_0(M)=\min\{n\in2\Z+1:\,n>\sqrt{t_+}\}$, subject only to the finite exceptional set
of Lemma~\ref{lem:finite}. Numerically $\sqrt{t_+}\approx0.76\sqrt N$ over
$4\le N\le1200$, consistent with the small values of $n$ observed in
Sec.~\ref{sec:numerics}. Oddness of $n$ plays no part in this sign argument, since
$A$ cannot see parity; it enters only through the spectral phase
condition~\eqref{eq:spectrum}.
\end{remark}

\begin{figure}[t]
\includegraphics[width=\columnwidth]{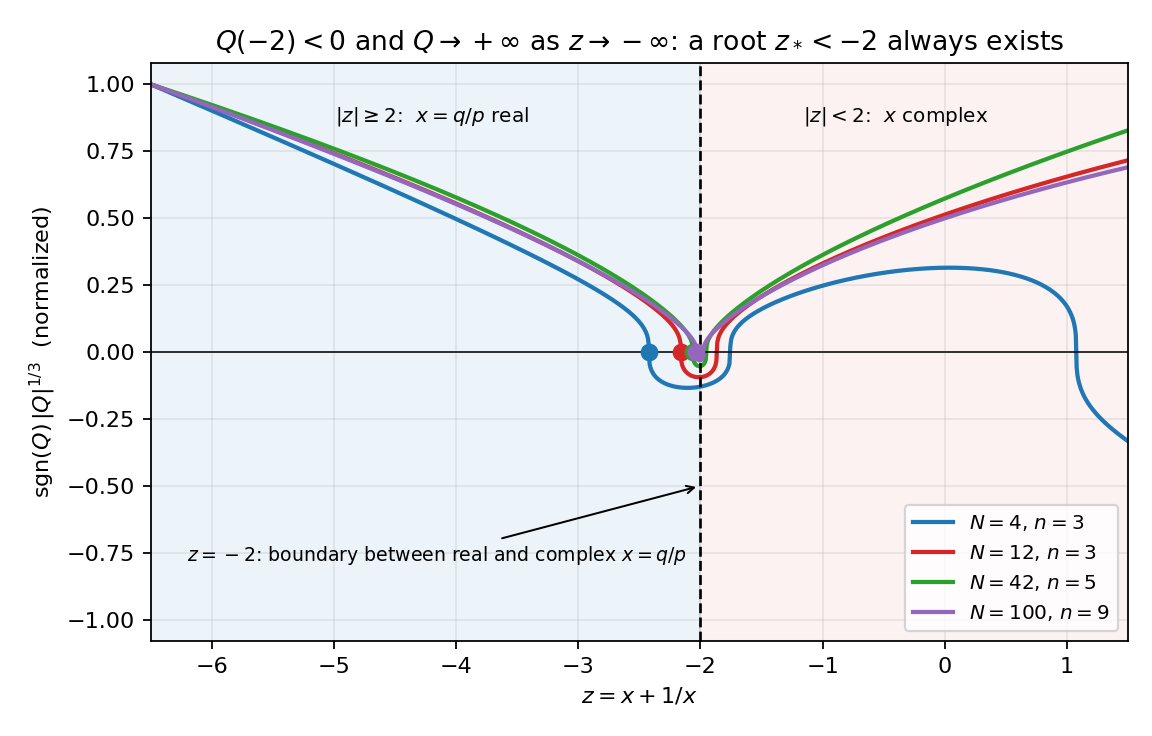}
\caption{\label{fig:cubic}%
The cubic $Q(z)$ for four system sizes, plotted as
$\mathrm{sgn}(Q)\,|Q|^{1/3}$ so that the crossing is visible against the cubic
growth, and normalised to unit maximum on the window shown. The dashed line marks
$z=-2$, the boundary between real and complex $x=q/p$; $x$ is real only in the
shaded regions $|z|\ge2$. In every case $Q(-2)<0$ by Lemma~\ref{lem:boundary}, while
$A<0$, which Lemma~\ref{lem:a4} supplies once $n$ is large enough, sends
$Q(z)\to+\infty$ as $z\to-\infty$. A real root $z_*<-2$ follows, and the markers
show it.}
\end{figure}

\section{Lifting the resultant root}

A root of the resultant is in general only a necessary condition for a common root.
Here that ambiguity can be removed by hand.

\begin{lemma}[Linear lift]
The subresultant chain of $F,G$ with respect to $y$ has degrees $[3,2,1,0]$. The
first nontrivial subresultant is linear, $\alpha y+\beta=0$, so whenever
$\alpha(x_*)\ne0$ the common root is unique and given by
$y_*=-\beta(x_*)/\alpha(x_*)$, a rational function of $x_*$ with real coefficients;
$y_*$ is therefore real whenever $x_*$ is.
\end{lemma}

\begin{lemma}[The exceptional set is finite]
\label{lem:finite}
For the family~\eqref{eq:family} and every $M\ge2$, $\mathrm{Res}_x(P,\alpha)$ is a
polynomial of degree $34$ in $n$ with nonvanishing leading coefficient. Hence, for
fixed $M$, it vanishes for at most $34$ values of $n$.
\end{lemma}

The polynomial $P$ is determined by the elimination only up to a real nonzero factor
$c(M,\gamma)$, and since $\deg_x\alpha=2$ such a rescaling multiplies
$\mathrm{Res}_x(P,\alpha)$ by $c^2>0$. Its degree in $n$, whether it vanishes
identically, whether it vanishes at a given $n$, and the sign of its leading
coefficient are therefore all normalisation-independent, even though the numerical
value of that coefficient is not. For the normalisation in which
$P$ is the primitive degree-six factor of $\mathrm{Res}_y(F,G)/\gamma^6$, the leading
coefficient is
\begin{equation}
2^{24}\gamma^{42}\bigl(M\gamma+2M+\gamma^{2}+3\gamma\bigr)^{12}
\bigl[(M-\gamma)^{2}+4M+2\gamma-3\bigr]^{6}>0 .
\end{equation}

Explicitly: for any $n$ with $\mathrm{Res}_x(P,\alpha)\ne0$, the polynomials $P$ and
$\alpha$ share no root, so no root $x$ of $P$ can satisfy $\alpha(x)=0$. The lift
$y_*=-\beta/\alpha$ is therefore well defined at $x_*$.

Since $A<0$ holds for all sufficiently large odd $n$, Lemma~\ref{lem:finite} excludes at most finitely many $n$, so an odd $n$ exists
satisfying both $A<0$ and $\alpha\ne0$.

\section{Main result}

\begin{theorem}
\label{thm:main}
For every $N\ge4$ there exists a real, time-independent, excitation-number-preserving
Hamiltonian of the form~\eqref{eq:model} with generally dense couplings, invariant
under $S_{N-2}$ acting on the initially unoccupied spins, and a finite time $t$,
such that
\begin{equation}
e^{-iHt}\ket{110\cdots0}=\Dket .
\end{equation}
\end{theorem}

The restriction $N\ge4$ arises because the $S_{N-2}$ reduction used here requires
$M=N-2\ge2$; the cases $N=2,3$ are outside the scope of this construction and are
not asserted to be impossible.

\begin{proof}
Set $M=N-2\ge2$. By Lemma~\ref{lem:a4}, $A_4<0$, so $A(n)\to-\infty$ as
$n\to\infty$; choose an odd $n$ with $A<0$, large enough also to avoid the finite
exceptional set of the previous lemma. Lemma~\ref{lem:boundary} gives
$Q(-2)<0$; with $A<0$ the intermediate value theorem gives a real root $z_*<-2$,
hence real $x_*=q/p$, and the linear lift gives real $y_*=a/p$. Rescaling so that $e_1=-n\pi$ fixes $p$; this is legitimate because
$e_1\ne0$ by Lemma~\ref{lem:e1}. Proposition~\ref{prop:stage1} determines the diagonal, and
Lemma~\ref{lem:realisable} converts it to physical $\epsilon_0,\epsilon_1,\epsilon_b,w$.
By construction $Hx_+=0$ and the spectrum on $x_+^\perp$ is $\{-n\pi,-\pi,\pi\}$, so
$e^{-iH}e_0=v$ by~\eqref{eq:transfer}. Undoing the scaling convention $t=1$ recovers
a general transfer time.
\end{proof}

\section{Worked example: $N=6$}
\label{sec:worked}

For $N=6$ ($M=4$) and $n=3$, the construction gives, at $t=1$,
\begin{equation}
\begin{aligned}
J_{01}&=-4.798441263, &
J_{0b}&=\phantom{-}1.257652965,\\
J_{1b}&=\phantom{-}0.521097052, &
J_{bc}&=\phantom{-}0.305202247,\\
\epsilon_0&=-0.040267117, &
\epsilon_1&=-1.419824076,\\
\epsilon_b&=-2.389154510, &&
\end{aligned}
\end{equation}
for all $b\ne c$ in $B=\{2,3,4,5\}$. This solution sits on the $z>+2$ family of Remark~\ref{rem:branch}, at
$z_*=2.82779\ldots$ and $x_*=q/p=0.414341\ldots$. I quote it in place of the $z<-2$
solution simply because its couplings are the smaller of the two. The reduced spectrum is exactly $(-3,-1,0,1)\pi$ and
$\bigl\|e^{-iH}e_0-v\bigr\|=5.2\times10^{-15}$ at full double precision; rebuilding
$H$ from the nine-digit values printed above gives $5.6\times10^{-10}$, the loss being
entirely the truncation. For comparison, the $z<-2$ solution at the same $N$ and $n$
has $J_{01}=1.483076621$, $J_{0b}=1.513472919$, $J_{1b}=-0.909515192$,
$J_{bc}=-3.008713387$, and $(\epsilon_0,\epsilon_1,\epsilon_b)=
(2.021228119,-2.516988260,5.413469046)$. All values are reproduced by the companion
notebook. Note $J_{bc}\ne0$ in both cases: the spectator sites must be coupled to one
another.

\begin{figure}[t]
\centering
\includegraphics[width=\columnwidth]{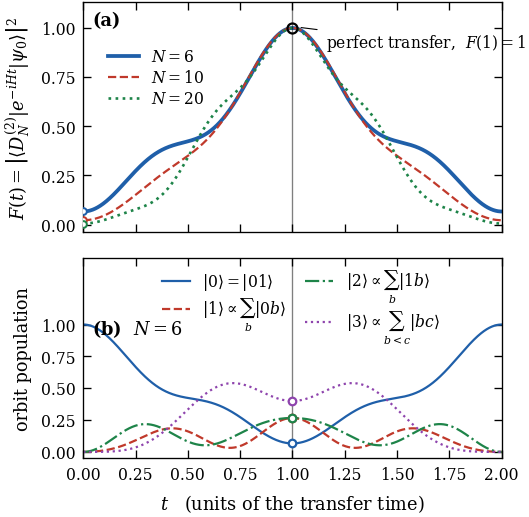}
\caption{\label{fig:dynamics}%
The transfer itself. I evaluate it in the full $\binom N2$-dimensional sector built
from the printed couplings, not in the reduced block, so that the reduction is tested
alongside the construction.
\emph{(a)}~$F(t)=|\braket{D_N^{(2)}|e^{-iHt}|\psi_0}|^2$ at $N=6,10,20$, where the
sector has dimension $15$, $45$ and $190$. At $t=0$ the fidelity is the bare overlap
$\binom N2^{-1}$; at $t=1$ it is $1$ to better than $10^{-12}$ in every case. Since
$e^{-iH}$ is $+1$ on $x_+$ and $-1$ on $x_+^\perp$, $e^{-2iH}=\mathbb 1$ and the
evolution is exactly $2$-periodic. \emph{(b)}~The four orbit populations at $N=6$.
Circles mark the Dicke amplitudes $(1,M,M,\binom M2)/\binom N2=(1,4,4,6)/15$, which
the four orbits reach together at $t=1$ and at no earlier time.}
\end{figure}

\section{Numerical validation}
\label{sec:numerics}

The proof is symbolic, so it deserves a check that shares nothing with it. I follow
the constructive route of Theorem~\ref{thm:main} at every $N$ from $4$ to $102$: for each $M=N-2$ the
smallest odd $n$ with $A<0$ is selected, a real root $x_*$ of $P$ is computed, $y_*=-\beta/\alpha$ is obtained from the linear subresultant, the couplings
are rescaled so that $e_1=-n\pi$, and the resulting Hamiltonian is diagonalised and
propagated. All $99$ cases succeed, with
\begin{equation}
\bigl\|e^{-iH}e_0-v\bigr\|\le1.1\times10^{-14},\qquad
\bigl\|Hx_+\bigr\|\le1.0\times10^{-15},
\end{equation}
and reduced spectra equal to $(-n,-1,0,1)\pi$ to within $10^{-14}$
(Fig.~\ref{fig:resid}). Theorem~\ref{thm:main} asks for an odd $n$ with both $A<0$ and $\alpha\ne0$. The
search imposes only the first condition and then tests the lift directly; in all $99$
cases the smallest odd $n$ with $A<0$ already missed the exceptional set, so $n$ never
had to be enlarged. Across this range the smallest admissible odd $n$ takes only
the values $3,5,7,9$ (Fig.~\ref{fig:scaling}).

\begin{figure}[t]
\includegraphics[width=\columnwidth]{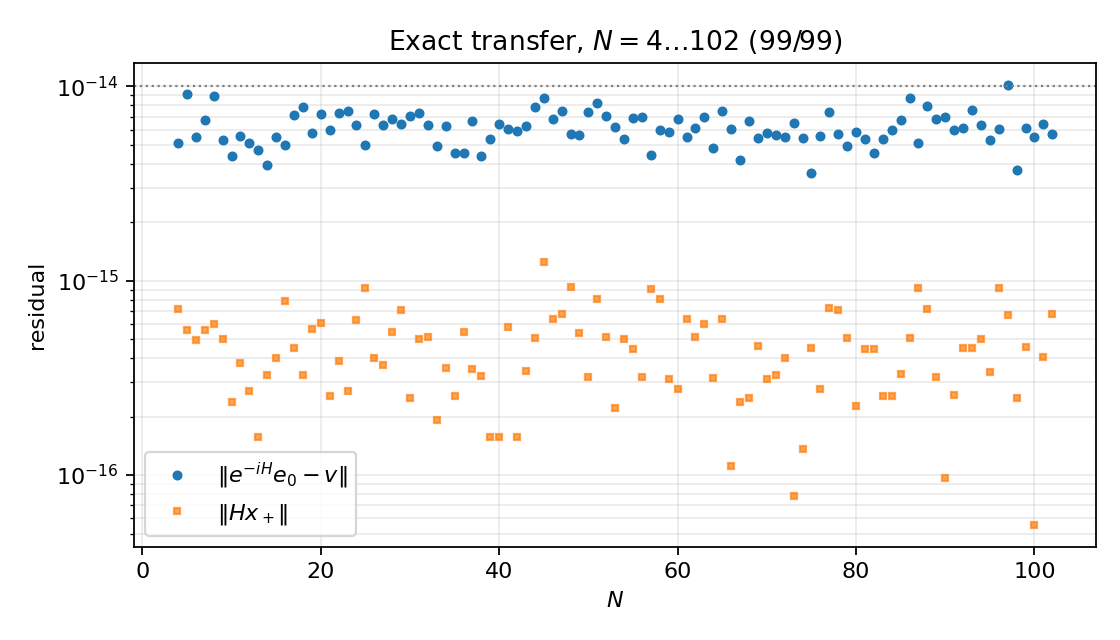}
\caption{\label{fig:resid}%
Transfer residual $\|e^{-iH}e_0-v\|$ and eigenvector residual $\|Hx_+\|$ for the
Hamiltonian produced by the construction of Theorem~\ref{thm:main}, for every
$N=4,\dots,102$. All $99$ cases succeed at machine precision. These are an
independent check of a symbolic construction, not evidence for the theorem.}
\end{figure} The same procedure
was also run at $N=502$, where $n=19$ and the residual is $9.9\times10^{-15}$. No
optimisation over Hamiltonians is performed at any point.

Theorem~\ref{thm:main} rests on three symbolic identities: the factorisation of
$P(-1)$ in Lemma~\ref{lem:boundary}, that of $A_4$ in Lemma~\ref{lem:a4}, and the
degree and leading coefficient of $\mathrm{Res}_x(P,\alpha)$ in
Lemma~\ref{lem:finite}. Each has been derived twice by independent means, once
through computer algebra and once through exact rational arithmetic using
hand-implemented polynomial and resultant routines that share no code with the
first. Both derivations, and every number quoted in this paper, are
reproduced by the companion notebook.

No numerical optimisation enters the construction, and this is not incidental.
Direct fidelity optimisation suits this problem badly: the exact-transfer solutions
form a thin algebraic subvariety of parameter space, and a gradient search can stall a
long way from an exact solution even where one exists. A control experiment on a
$20$-dimensional sector makes the point. Targets built to be exactly reachable were
recovered only to $\|\,\cdot\,\|\sim10^{-4}$, which is nowhere near enough to tell a
reachable target from an apparently unreachable one, so a residual at that level says
nothing either way. Everything reported here comes instead from exact polynomial
elimination, with propagation used only to verify it.

\section{Resource scaling and error sensitivity}
\label{sec:scaling}

Theorem~\ref{thm:main} asserts existence and stops there; it says nothing about the
price of the construction as $N$ grows. Since the couplings come out constructively,
that price can simply be measured. Table~\ref{tab:scaling} records it along the branch
the proof uses, taking at each $N$ the smallest odd $n$ with $A<0$.

\begin{table*}[t]
\caption{\label{tab:scaling}%
Resource scaling of the constructed Hamiltonians. $n$ is the smallest odd integer
with $A<0$; $\max|J|$ is over the four coupling classes $a,p,q,w$; the on-site spread
is $\max_i\epsilon_i-\min_i\epsilon_i$, which is the gauge-invariant quantity since a
uniform shift of all $\epsilon_i$ acts as a global phase in a fixed-excitation sector.
$\Delta E$ is the energy standard deviation in the initial state $\ket{\psi_0}$ and
$\theta_N=\arccos\binom N2^{-1/2}$ is the Mandelstam--Tamm angle; at $t=1$ the bound
reads $\Delta E\ge\theta_N$, so the last column is the factor by which the transfer
time exceeds the Mandelstam--Tamm limit.}
\begin{ruledtabular}
\begin{tabular}{rrrrrrr}
$N$ & $n$ & $\max|J|$ & $\max\epsilon-\min\epsilon$ & $\Delta E$ & $\theta_N$ & $\Delta E/\theta_N$\\
\colrule
   4 &  3 &  2.128 &      8.0 & 3.403 & 1.1503 & 2.959\\
   6 &  3 &  3.009 &      7.9 & 3.531 & 1.3096 & 2.697\\
  10 &  3 &  3.637 &     22.3 & 3.501 & 1.4212 & 2.463\\
  20 &  5 &  7.243 &    117.4 & 4.290 & 1.4982 & 2.863\\
  40 &  5 &  7.486 &    273.1 & 3.922 & 1.5350 & 2.555\\
  80 &  7 & 10.765 &    824.6 & 4.149 & 1.5530 & 2.671\\
 102 &  9 & 13.941 &   1374.6 & 4.520 & 1.5569 & 2.903\\
 200 & 11 & 17.146 &   3372.5 & 4.340 & 1.5637 & 2.776\\
 400 & 17 & 26.614 &  10557.8 & 4.645 & 1.5673 & 2.964\\
 800 & 23 & 36.066 &  28734.9 & 4.581 & 1.5690 & 2.920\\
1200 & 27 & 42.361 &  50695.3 & 4.491 & 1.5696 & 2.861\\
\end{tabular}
\end{ruledtabular}
\end{table*}

\begin{figure*}[t]
\centering
\includegraphics[width=\textwidth]{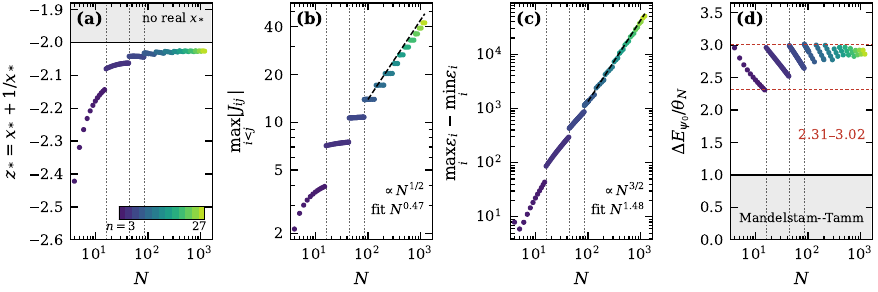}
\caption{\label{fig:scaling}%
The construction run at every $N$ from $4$ to $102$ and then on a geometric grid out
to $N=1200$, all on the $z<-2$ branch. Colour encodes the spectral integer $n$, and
the dotted verticals mark the three places below $N=102$ where $n$ steps; those steps
are where the jumps in (b) and (c) come from. \emph{(a)}~The root $z_*$ stays below
$-2$ throughout, as Sec.~\ref{sec:exist} requires, approaching the boundary from
below as $N$ grows. Above $-2$ no real $x_*$ exists. \emph{(b)}~The largest coupling
and \emph{(c)}~the on-site spread, against $N^{1/2}$ and $N^{3/2}$ guides. The gap
between these two exponents is what Sec.~\ref{sec:scaling} is about: the couplings
stay mild, the detunings do not. \emph{(d)}~The Mandelstam--Tamm ratio
$\Delta E_{\psi_0}/\theta_N$, which the bound keeps at or above $1$. It falls across
each plateau of fixed $n$ and jumps back at every step, staying between $2.31$ and
$3.02$ over the whole range.}
\end{figure*}

Three things happen at once here, and they pull in different directions.

\paragraph{The couplings grow slowly.}
Fitting a power law to the sweep of Fig.~\ref{fig:scaling} over $N\ge100$ gives
$n\sim N^{0.46}$ and $\max|J|\sim N^{0.47}$. The exponents are window-dependent at
the second decimal, since $n$ is a step function of $N$ and every step displaces the
whole curve: restricting the fit to the rows of Table~\ref{tab:scaling} alone returns
$0.51$ for $\max|J|$. Read $\sqrt N$ into them and nothing finer. The spectral
spread on the four-dimensional sector is $n\pi$ by construction, so it too grows as
$\sqrt N$, consistent with Remark~\ref{rem:threshold} and its prediction
$n_0\approx0.76\sqrt N$.

\paragraph{The on-site range grows much faster.}
The same fit gives $\max\epsilon-\min\epsilon\sim N^{1.48}$, so the ratio of on-site
spread to exchange coupling climbs roughly linearly in $N$: the Hamiltonians become
steadily more \emph{detuned}. This is the dominant static resource requirement, well
ahead of coupling strength, and it is what would stand in the way of an
implementation, since it sets both the dynamic range and the absolute calibration
accuracy demanded of the on-site control.

\paragraph{The transfer time remains near the speed limit.}
Since $|\braket{\psi_0|D_N^{(2)}}|=\binom N2^{-1/2}\to0$, the rotation angle
$\theta_N\to\pi/2$, and the Mandelstam--Tamm bound
$t\,\Delta E_{\psi_0}\ge\theta_N$ forces $\Delta E\gtrsim\pi/2$ at $t=1$ however the
Hamiltonian is chosen. The measured $\Delta E$ stays between $3.4$ and $4.7$ across
$4\le N\le1200$ (and equals $4.51$ at $N=5000$), so
\begin{equation}
\frac{t}{t_{\mathrm{MT}}}=\frac{\Delta E_{\psi_0}}{\theta_N}\in[2.31,3.02]
\end{equation}
over that whole range. The ratio is not monotone: within each plateau of fixed $n$ it
falls, and at each step of $n$ it jumps back up (Fig.~\ref{fig:scaling}(d)). The
extremes quoted are attained at $N=15$, the last size before the step to $n=5$, and at
$N=85$, just before the step to $n=9$; the sparse sizes of Table~\ref{tab:scaling}
straddle the first of these and do not show it. So the transfer runs within an $O(1)$ factor of the Mandelstam--Tamm lower bound at
every size, even as the on-site range of the Hamiltonian grows. There is no tension
between the two: large on-site energies are a static resource requirement, whereas the
bound constrains $\Delta E_{\psi_0}$, and those are different quantities.

\subsection{Correlated versus independent disorder}
\label{sec:disorder}

Fixing seven parameters exactly invites the obvious objection: a measure-zero
solution ought to be fragile. Perturbing it tells a more structured story. What
governs the damage is not the size of an error so much as whether it is correlated
across bonds.

We compare two ensembles of relative Gaussian perturbations at the same strength
$\delta$:
\begin{itemize}
\item[(A)] \emph{$S_M$-preserving}: one draw per coupling \emph{class}, applied
uniformly to all bonds of that class --- systematic calibration drift;
\item[(B)] \emph{$S_M$-breaking}: an independent draw for every individual bond
$J_{ij}$ and every site energy $\epsilon_i$ --- microscopic fabrication scatter.
\end{itemize}
Table~\ref{tab:disorder} gives the median transfer fidelity over $40$ realisations.

\begin{table}[b]
\caption{\label{tab:disorder}%
Median transfer fidelity $|\braket{D_N^{(2)}|e^{-iH}|\psi_0}|^2$ over $40$
realisations of relative Gaussian disorder of strength $\delta$, for the
$S_M$-preserving ensemble (A) and the $S_M$-breaking ensemble (B).}
\begin{ruledtabular}
\begin{tabular}{rrrrrrr}
 & \multicolumn{3}{c}{(A) $S_M$-preserving} & \multicolumn{3}{c}{(B) $S_M$-breaking}\\
\colrule
$N$ & $1\%$ & $3\%$ & $10\%$ & $1\%$ & $3\%$ & $10\%$\\
\colrule
 6 & 0.9988 & 0.9880 & 0.8885 & 0.9997 & 0.9967 & 0.9541\\
10 & 0.9920 & 0.9468 & 0.5912 & 0.9997 & 0.9961 & 0.9538\\
20 & 0.8239 & 0.2887 & 0.0322 & 0.9985 & 0.9865 & 0.8480\\
\end{tabular}
\end{ruledtabular}
\end{table}

\begin{figure}[t]
\centering
\includegraphics[width=\columnwidth]{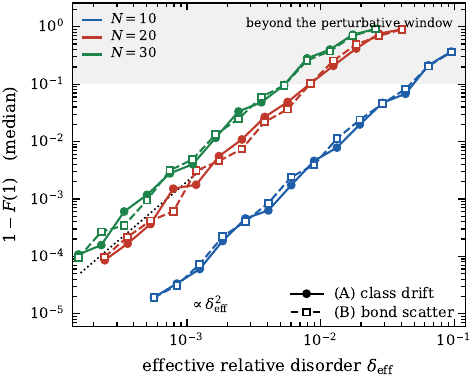}
\caption{\label{fig:disorder}%
Median infidelity under relative Gaussian disorder confined to the
spectator--spectator class $w$, for the two ensembles of Sec.~\ref{sec:disorder}: one
draw per class (A, filled circles) and one draw per bond (B, open squares), with
$200$ realisations behind each point. Ensemble B is plotted against
$\delta_{\mathrm{eff}}=\delta/\sqrt{\binom M2}$ and ensemble A against $\delta$
itself, so~\eqref{eq:collapse} predicts that the two should land on top of each
other. They do, over three decades and at all three sizes, with the slope of $2$ a
perturbative unitary error demands. Inside the shaded band that picture has broken
down and the ensembles come apart, which is the first qualification discussed below.}
\end{figure}

The ordering is the opposite of what the measure-zero objection anticipates:
\emph{breaking} the symmetry at random is the benign direction, and the gap widens
with $N$. Resolving the perturbation by coupling class locates the cause. At
$\delta=1\%$ and $N=30$, disorder confined to $a$ gives fidelity $0.99996$, to $p$
gives $0.9997$ (A) and $1.0000$ (B), to $q$ gives $0.9995$ (A) and $1.0000$ (B),
while disorder confined to the spectator--spectator coupling $w$ gives $0.4828$ (A)
against $0.9987$ (B). Essentially the entire sensitivity is a systematic error on
$w$.

The mechanism is visible in~\eqref{eq:diag}: $w$ enters the reduced block only
through the diagonal, and it arrives there with multiplicity: $(M-1)w$ in $d_1$ and
$d_2$, and $2(M-2)w$ in $d_3$. A class-wide relative error $\delta$ is therefore accumulated
coherently, $\propto M\delta$, whereas independent errors on the $\binom M2$
individual bonds $J_{bc}$ enter only through their sample mean. This predicts a
quantitative relation between the two ensembles: bond scatter of strength $\delta$
should act like a systematic class error of strength
\begin{equation}
\delta_{\mathrm{eff}}=\frac{\delta}{\sqrt{\binom M2}}
=\delta\,\sqrt{\frac{2}{M(M-1)}}\;\sim\;\frac{\sqrt2\,\delta}{M} .
\label{eq:collapse}
\end{equation}
The prediction holds. For $w$-only disorder at $N=20$ ($\sqrt{\binom M2}=12.37$) the
ensemble-(B) fidelities at $\delta=0.01,0.03,0.10$ are $0.99902$, $0.99363$,
$0.91073$, against ensemble-(A) fidelities at $\delta_{\mathrm{eff}}$ of $0.99917$,
$0.99311$, $0.89383$; at $N=30$ ($\sqrt{\binom M2}=19.44$) the pairs are
$0.99903/0.99916$ and $0.99362/0.99011$. The infidelity is quadratic in $\delta$, with
fitted exponents $1.82$, $1.98$ and $2.01$ at $N=10,20,30$, as expected for a
perturbative unitary error. We do not fit an exponent in $M$: the smallest admissible
odd $n$ is a step function of $N$, and that discreteness contaminates any such fit.

Two qualifications go with this. First, the collapse~\eqref{eq:collapse} is
perturbative, and outside that window it errs the wrong way: at $\delta=0.3$ and
$N=20$ the ensemble-(B) fidelity is $0.288$ against a rescaled ensemble-(A)
prediction of $0.459$, so the self-averaging estimate \emph{under}-predicts the
damage. Second, every comparison here is internal to this construction. I make no
claim about how it fares against other engineered PST schemes; settling that would
need a common disorder model and a common fidelity metric.

None of this is a robustness claim. It is a statement about which errors matter:
\begin{equation}
\text{systematic class drift: }\delta_{\mathrm{eff}}\sim\delta,
\qquad
\text{bond scatter: }\delta_{\mathrm{eff}}\sim\delta/M .
\end{equation}
Exact $S_M$ symmetry is thus an analytical device for the reduction, and not a demand
that every physical spectator bond be identical. Zero-mean bond-to-bond variation gets
suppressed by the spectator multiplicity. What has to be held accurately is the mean
value of the spectator--spectator coupling.

\section{Relation to existing theory}
\label{sec:krawtchouk}

\paragraph{Spectral PST.}
Strong cospectrality and the associated spectral projection conditions are
established concepts~\cite{Godsil2012,Coutinho2014}, and PST between
general real pure states has been characterised in that
language~\cite{GKM2025}. This framework is used here as machinery; no new spectral
criterion is claimed.

Reference~\cite{GKM2025} also handles constrained Hamiltonians, namely adjacency and
Laplacian matrices of graphs, and classifies which real pure states admit PST on
complete graphs among others. That classification does not cover the present result.
There the Hamiltonian is an adjacency or Laplacian matrix on the vertex set itself,
and the claim is that \emph{some} pair of real pure states admits PST on a given
connected graph. Here the operator is a two-excitation sector Hamiltonian for $N$
spins, its entries are tied together by the ratio
constraints~\eqref{eq:constraints}, and the pair of states is fixed in advance. What
I add is the realisation of such a transfer inside the constrained
family~\eqref{eq:model} with $S_{N-2}$ symmetry.

\paragraph{Dicke-state preparation.}
Existing routes to $\ket{D_N^{(k)}}$ fall into three families, none of which is a
static-Hamiltonian transfer. Circuit constructions prepare the state with
$O(kn)$ gates and no ancillas~\cite{Bartschi2019,Bartschi2022}, including
divide-and-conquer variants adapted to sparse hardware
connectivity~\cite{Aktar2022}. Measurement-and-feedback schemes reach
$\ket{D_N^{(k)}}$ from a product state using global rotations and collective $J_z$
measurements in polylogarithmic depth~\cite{Yu2024}. Physical schemes based on
dissipation, adiabatic passage, or time-dependent global control prepare Dicke states
in spin and atomic ensembles~\cite{Luo2011,Stojanovic2023}. Each of these needs a gate sequence, a projective measurement, or time-dependent
control. This construction needs none of them. The Hamiltonian is fixed, the evolution
is free, and the state simply appears at a prescribed time.

\paragraph{Krawtchouk and $\mathrm{SU}(2)$ chains.}
Engineered PST chains with couplings $J_k\propto\sqrt{k(N-k)}$ are a standard
precedent~\cite{Christandl2004}. The reduced matrix~\eqref{eq:Heff} contains
$\sqrt M$ and $\sqrt{2(M-1)}$ and may superficially resemble these. The state geometry is different, though. A Krawtchouk chain keeps the relevant
states inside a permutation-symmetric sector; here $\ket{110\cdots0}$ is localised and
unsymmetric while $\Dket$ is fully symmetric. Only $S_{N-2}$ survives, and the
four-dimensional problem it leaves behind still distinguishes the two occupied sites,
which Sec.~II showed to be indispensable.

\paragraph{Long-range couplings.}
Long-range extended XY interactions have been shown to speed up site-to-site
transfer relative to the short-range case, reducing the time needed to exceed the
classical fidelity threshold~\cite{Ahuja2026}. This construction also rests on a non-local interaction graph, with the spectator
block internally coupled and no locality imposed anywhere. The task is different,
however: exact transfer to a symmetric multipartite target, not transfer between two
sites.

\section{Scope and limitations}
\label{sec:scope}

\begin{enumerate}
\item \textbf{Two excitations only.} Nothing here establishes the corresponding
result for $\ket{D_N^{(k)}}$ with $k\ge3$, and the mechanism does not extend. The
$S_{N-k}$ orbits of $k$-subsets are labelled by $S\subseteq A$ subject to
$k-|S|\le N-k$, so their number is $\sum_{j\ge\max(0,2k-N)}\binom kj$, which equals
$2^k$ exactly when $N\ge2k$. The diagonal of the reduced matrix is set by the $k+2$
parameters $\epsilon_0,\dots,\epsilon_{k-1}$, $\epsilon_b$ and $w$: writing
$r=k-|S|$ and $M=N-k$, the orbit labelled by $S$ has diagonal entry
$\sum_{i\in S}\epsilon_i+r\,\epsilon_b+r(M-r)\,w$, the last factor counting the
spectator hops available from any state in the orbit. The corresponding linear map has
rank exactly $k+2$ whenever the orbit count exceeds $k+2$, so for $N\ge2k$ it is onto
precisely when $2^k\le k+2$. Stage-I freedom means the diagonal absorbing $Hx_+=0$
unaided, which is what leaves $p,q,a$ free for the spectral problem; for $N\ge2k$ it
therefore holds exactly when $k\le2$, with $k=2$ the saturating case $4=4$. (It survives in the isolated edge case $N=k+1$, where the orbit count collapses to
$k+1$; by particle--hole complement that case is a single-hole problem, and so is
qualitatively different from the generic regime $N\ge2k$.) The downstream equation count grows likewise: the family carries
$\binom k2+2k+2$ parameters against $2^{k+1}-2$ raw equations ($2^k$ from
$Hx_+=0$, and $2^k-2$ from fixing the remaining $2^k-1$ eigenvalues up to an overall
scale), a shortfall of
$3,14,40$ at $k=3,4,5$. I do not claim these equations are independent, so read the count as an indication
and not as a dimensional obstruction. Both statements block the architecture used
here, meaning the $S_{N-k}$ reduction together with the $x_+$ zero-mode condition.
Neither is a no-go for static-Hamiltonian generation of higher-excitation Dicke
states: a different symmetry group, or a construction that never routes through
$x_+$, is under no obligation to obey them.

A second ingredient, independent of the first, would also have to be replaced.
Dropping the sextic to a cubic is no algebraic accident; it is the image of a physical
symmetry. Exchanging the two initially occupied sites sends $p\leftrightarrow q$
by~\eqref{eq:exchange}, therefore $x\leftrightarrow x^{-1}$, therefore the palindromic
$P(x)=x^{6}P(1/x)$, and $z=x+x^{-1}$ is precisely the invariant of that $\Z_2$ action.
At $k\ge3$ the occupied block carries a full $S_k$ action instead of a two-element
exchange. Nothing in the present argument determines the corresponding invariant
parametrisation of the coupling ratios, or whether it buys a comparable drop in
degree. Identifying the $S_k$-invariant replacement for
$z=x+x^{-1}$ is therefore a concrete question separate from the Stage-I counting
above, and both would have to be answered for the method to carry over.
\item \textbf{Hamiltonian engineering, not distribution.} The construction assumes
globally engineered couplings and on-site energies, and so does not constitute an
LOCC entanglement-distribution protocol. Distribution of Dicke states over quantum
networks is a separate and well-studied problem, addressed in the optical setting by
multipartite twin-field interference with heralded detection~\cite{Roga2023} and by
deterministic schemes over collective-noise channels~\cite{Wang2016}; these use
quantum channels, not local operations alone. The present work concerns state
generation by fixed spin-network dynamics. In passing, $\ket{D_N^{(2)}}$ has Schmidt rank at most $3$ across every bipartition,
so what is hard here is dynamical and has nothing to do with the entanglement rank of
the target.
\item \textbf{Static resource versus error sensitivity.} These are distinct and
scale differently (Sec.~\ref{sec:scaling}). The static requirement is on-site dynamic range, growing as $N^{1.48}$ against
couplings that grow only as $N^{0.47}$, and that is the principal practical
limitation. The error requirement is a separate matter. The dominant sensitivity I
tested is systematic drift in the spectator--spectator coupling class, whose tolerance
tightens as $1/M$, whereas independent bond-to-bond scatter is suppressed as $1/M$ and
therefore becomes \emph{more} tolerable as $N$ grows. A parameter being large does not
by itself make its fabrication scatter dangerous, and the two requirements should not
be run together. Neither statement has been tested outside the
perturbative regime or against other PST constructions.

\item \textbf{No optimality.} I do not claim the seven-parameter family, the
$S_{N-2}$ symmetry, or the spectral family $(-n,-1,1)$ is minimal or optimal.
\item \textbf{No minimal-time claim.} $t=1$ is a scaling convention.
\item \textbf{No claim for general targets.} Unconstrained real-state PST is already
known; the contribution here is the constrained realisation.
\end{enumerate}

\section{Conclusion}

I have constructed a static, real, excitation-preserving spin-network Hamiltonian
that perfectly transfers a localised two-excitation state to the two-excitation Dicke
state for every $N\ge4$. The construction proceeds by $S_{N-2}$ reduction to a
four-dimensional space, an unconstrained Stage-I choice making
$\tfrac12(e_0+v)$ a zero eigenvector, and an inverse spectral problem that reduces to
a reciprocal sextic and thence to a cubic. Existence follows from the sign of a single
factored coefficient together with a boundary evaluation, with no numerical input;
because that coefficient turns out to be a quadratic in $n^{2}$, the choice of $n$
comes with an explicit threshold rather than merely an asymptotic guarantee
(Remark~\ref{rem:threshold}).

What holds for every $N\ge4$ is \emph{existence}, and with it an algorithm: solve a
cubic, then lift through a linear subresultant. That is not a closed-form family of
couplings $J_{ij}(N)$, and I claim no such form. Costing the route shows the couplings
and the spectral spread growing only as $\sqrt N$ while the on-site range grows as
$N^{1.48}$, the transfer time staying within roughly a factor of three of the
Mandelstam--Tamm bound at every size, and the error sensitivity concentrated in
systematic drift of a single
coupling class rather than in microscopic disorder, which self-averages
(Sec.~\ref{sec:scaling}). Whether a different architecture extends the result to
$k\ge3$ remains open; Sec.~\ref{sec:scope} identifies two separate obstructions that such an
architecture would have to circumvent.

\begin{acknowledgments}
The author thanks Tushar (Max Planck Institute of Quantum Optics) for valuable
discussions. All results and conclusions are solely the responsibility of the
author.
\end{acknowledgments}

\end{document}